%% file: faultShort.tex
\documentclass[a4paper,USenglish]{article}
\usepackage[T1]{fontenc}
\usepackage{float}
\usepackage{amsthm,amsmath,amssymb,amsfonts,authblk,hyperref,setspace}
\usepackage{tikz,xcolor}
\usetikzlibrary{shapes,arrows,backgrounds,calc,trees}

\pgfdeclarelayer{background}
\pgfsetlayers{background,main}
\input{macros}
\usepackage{fullpage}
\newcommand{\comment}[1]{}
\newcommand{\tildeO}{\widetilde{O}}
\input{style}
\SetKwInOut{Input}{Input}
\SetKwInOut{Output}{Output}
\begin{document}

	
	
\title{Near Optimal Sized Weight Tolerant Subgraph for Single Source Shortest Path\footnote{The research leading to these results has received funding from the European Research Council under the European Union's Seventh Framework Programme (FP/2007-2013)/ERC Grant Agreement n. 616787.}}

\author[1]{Diptarka Chakraborty\thanks{diptarka@iuuk.mff.cuni.cz}}
\author[2]{Debarati Das\thanks{debaratix710@gmail.com}}
  \affil[1,2]{Computer Science Institute of Charles University,
		Malostransk{\'e}  n{\'a}mesti 25,
		118 00 Praha 1, Czech Republic}

	\maketitle
	\begin{abstract}
	In this paper we address the problem of computing a sparse subgraph of a weighted directed graph such that the exact distances from a designated source vertex to all other vertices are preserved under bounded weight increment. Finding a small sized subgraph that preserves distances between any pair of vertices is a well studied problem. Since in the real world any network is prone to failures, it is natural to study the fault tolerant version of the above problem. Unfortunately, it turns out that there may not always exist such a sparse subgraph even under single edge failure [Demetrescu \emph{et al.} '08]. However in real applications it is not always the case that a link (edge) in a network becomes completely faulty. Instead, it can happen that some links become more congested which can easily be captured by increasing weight on the corresponding edges. Thus it makes sense to try to construct a sparse distance preserving subgraph under the above weight increment model where total increase in weight in the whole network (graph) is bounded by some parameter $k$. To the best of our knowledge this problem has not been studied so far.
	
	In this paper we show that given any weighted directed graph with $n$ vertices and a source vertex, one can construct a subgraph that contains at most $e \cdot (k-1)!2^kn$ many edges such that it preserves distances between the source and all other vertices as long as the total weight increment is bounded by $k$ and we are allowed to have only integer valued (can be negative) weight on each edge and also weight of an edge can only be increased by some positive integer. Next we show a lower bound of $c\cdot 2^kn$, for some constant $c \ge 5/4$, on the size of the subgraph. We also argue that restriction of integer valued weight and integer valued weight increment are actually essential by showing that if we remove any one of these two restrictions we may need to store $\Omega(n^2)$ edges to preserve distances.
	\end{abstract}
	\section{Introduction}
	In the real world, networks are prone to failures and most of the time such failures are unpredictable as well as unavoidable in any physical system such as communication network or road network. For this reason, in the recent past, researchers study many graph theoretic questions like connectivity~\cite{PP13, PP14, Parter15, BCR15, BCR16, Choudhary16}, finding shortest distance~\cite{DTCR08}, building data structures that preserves approximate distances~\cite{Luk99, LNS02, CZ04, CLPR09, DK11, BGLP14, BGGLP15, BCPS15, BGPW17} etc. under the fault tolerant model. Normally such failures are much smaller in number comparative to the size of the graph. Thus we can associate a parameter to capture the number of edge or vertex failures and try to build fault tolerant data structures of size depending on this failure parameter for various graph theoretic problems.
	
	Unfortunately, in case of single source shortest path problem, it is already known from~\cite{DTCR08} that there are graphs with $n$ vertices for which to preserve the distances under even single edge failure, we need to store a subgraph of size at least $\Omega(n^2)$. On the other hand, in case of reachability problem we know the construction of connectivity preserving subgraph of size only $O(2^k n)$ due to~\cite{BCR16} where $k$ is the number of edge failures. However, in case of real world applications, it is not always the case that there are failures of edges or vertices. Instead, for weighted graphs, weight of any edge or vertex can be increased. For simplicity, we consider weight to be only on the edges of the graph. In general, weight of an edge captures aspect like congestion on a particular link in a network. So it is quite natural to consider the scenario when some links (or edges) become more congested. Again the good thing is that most of the time such congestion is bounded, i.e., over a network total increase in congestion is bounded because of many reasons like bounded maximum number of consumers present in a network at any particular time. O	ne can easily capture the increase in congestion by a parameter $k$ that bounds the amount of increase in weight of edges over the whole graph. Occurrence of such bounded congestion motivates us to study the single source shortest path problem under this model.
	
	In spite of being an appealing model, to the best of our knowledge this bounded congestion increase model or bounded weight increment model has not been studied so far. In this paper we initiate the study of single source shortest path problem in weighted directed graphs in the above model. The main goal is to find a sparse subgraph that preserves distance between a designated source and any other vertices under weight increment operation. We formally define such subgraph below.
	\begin{definition}
	Given a graph $G$ along with weight function $w$, a source vertex $s \in V(G)$ and an integer $k \ge 1$, a subgraph $H=(V(G),E')$ where $E' \subseteq E(G)$ is said to be {\em $k$-Weight Tolerant Shortest-path Subgraph ($k$-WTSS)} of $G$ if for any weight increment function $I:E(G) \to \mathbb{N}$ such that $\sum_{e \in E(G)} I(e)\le k$, the following holds: for the weight function defined by $w'(e)=w(e)+I(e)$ for all $e \in E(G)$, 
	$$dist_{G,w'}(s,t)=dist_{H,w'}(s,t) \text{ for any } t \in V(G).$$
	\end{definition}
	Though in the above definition we restrict ourselves to an increment function whose range is $\mathbb{N}$, one can naturally extend the definition to any range of increment functions. However, if we take the range to be rational numbers then there may not exist any sparse $k$-WTSS even for $k=1$ (see Section~\ref{sec:rationalincrementLB}). This is the reason why we consider such restriction on increment function in the above definition.
	
	Single source shortest path is one of the most fundamental problems in Graph Theory as well as in computer science. Thus construction of sparse $k$-WTSS is interesting from both the theoretical and practical point of view. One can also view this problem of finding $k$-WTSS as a generalization of finding $k$-Fault Tolerant (single source) Reachability Subgraph ($k$-FTRS) for which optimal solution is known due to~\cite{BCR16}. If in the given graph one assigns zero weight to each of the edges, then any $k$-WTSS of that graph will also be a $k$-FTRS of that graph. This is because we can view each edge fault as incrementing weight by one and then it is easy to see that for any vertex $t$ there exists an $s-t$ path if and only if shortest distance between $s$ and $t$ is zero under this reduction. This fact also motivates the study of constructing $k$-WTSS because $k$-FTRS has several applications like fault tolerant strong-connectedness~\cite{BCR16}, dominators~\cite{LT79, BCR16}, double dominators~\cite{TD05} etc.
	
	The main contribution of this paper is to provide an efficient construction of a sparse $k$-WTSS for any $k \ge 1$, where sparsity of $k$-WTSS depends on the parameter $k$.
	
	\begin{theorem}
	\label{thm:main}
	There exists an $O((k)^k m^2 n)$-time algorithm that for any given integer $k \ge 1$, and a given directed graph $G$ with $n$ vertices and $m$ edges along with a weight function $w:E(G) \to \mathbb{Z}$ and a source vertex $s$, constructs a $k$-WTSS of $G$ with at most $O_k(n)$ edges where $O_k(\cdot)$ notation denotes involvement of a constant that depends only on the value of $k$. Moreover, in-degree of every vertex in this $k$-WTSS will be bounded by $e\cdot (k-1)! 2^k$.
	\end{theorem}
	Next, we prove a lower bound of $c \cdot 2^k n$ for some constant $c \ge 5/4$, on the size of $k$-WTSS.
	\begin{theorem}
	\label{thm:LB}
	For any positive integer $k \ge 2$, there exists an positive integer $n'$ such that for all $n > n'$, there exists a directed graph $G$ with $n$ vertices and a weight function $w:E(G) \to \mathbb{Z}$, such that its $k$-WTSS must contain $c \cdot 2^k n$ many edges for some constant $c \ge 5/4$.
	\end{theorem}
	Note that as we have previously argued that the construction of $k$-WTSS implies a construction of $k$-FTRS, so $2^kn$ lower bound on the size of a $k$-FTRS due to~\cite{BCR16} also directly gives the same lower bound on size of a $k$-WTSS. In the above theorem we slightly strengthen that lower bound by a constant factor for our problem.
	
	We also show that considering rational valued (instead of integer valued as in the above two theorems) weight function or rational valued weight increment function makes the problem of finding sparse $k$-WTSS impossible. More specifically, we show that in both the cases there are graphs with $n$ vertices for which any $k$-WTSS must be of size at least $c\cdot n^2$ for some constant $c  > 0$ even for $k=1$ (see Section~\ref{sec:generalLB}).
	
	One can farther relax our model by also allowing decrement operation on edge weight. Weight decrement is also natural in real world application because for any network it is possible that some links become less congested. Unfortunately, one can easily show that there are graphs for which there is no sub-quadratic sized subgraph that preserves the distances from a single source under this relaxed model. Readers may refer to Appendix~\ref{app:decrementLB} for the details.
	
	\subsection{Related works}
	Single source shortest path is a well studied problem under the edge or vertex failure model. Similar to our definition of $k$-WTSS, one can easily define $k$-Fault Tolerant Shortest-path Subgraph ($k$-FTSS) that preserves the distance information from a specific source vertex under at most $k$ edge failures. Unfortunately, we know that there are weighted graphs for which no sparse $k$-FTSS exists even for $k=1$, i.e., there are weighted graphs with $n$ vertices for which any $1$-FTSS must contain $\Omega(n^2)$ many edges~\cite{DTCR08}. This lower bound on size of $1$-FTSS is even true for undirected graphs. However, better bounds are known for unweighted graphs for $k\le 2$. Parter and Peleg~\cite{PP13} provided a construction of $O(n^{3/2})$ sized $1$-FTSS and showed that this bound is optimal. Later, Parter~\cite{Parter15} extended the construction to the case $k=2$ for undirected graphs on the cost of weakening the bound. He gave an algorithm to compute $2$-FTSS of size $O(n^{5/3})$ along with a matching lower bound. So far there is nothing known for the case $k > 2$.
	
	However, the situation is much better for single source reachability problem which is closely related to single source shortest path problem. Baswana, Choudhary and Roditty~\cite{BCR16} showed that we can compute $k$-FTRS, which is a subgraph that preserves the reachability information from a given source under at most $k$ edge failures, containing $2^kn$ many edges. They also provided a matching lower bound. We have already argued that computing $k$-FTRS can be reduced to computing $k$-WTSS and thus it is natural to ask whether a similar result also holds for $k$-WTSS. Another interesting related problem is to compute fault tolerant reachability oracle. It is trivial to see that using $O(2^kn)$ size $k$-FTRS~\cite{BCR16} one can answer any reachability query in $O(n)$ time for any constant value of $k$. However for $k \le 2$, $O(n)$ size data structure is known that can answer any single source reachability query in $O(1)$ time~\cite{LT79, Choudhary16}. Very recently, existence of an efficient algorithm to find strongly connected component in fault tolerant model has also been shown~\cite{BCR17}.
	
	Now let us come back to the shortest path problem. Instead of preserving the exact distances (between any pair of vertices), if we consider to preserve the distances only approximately, then much better results are known. In the literature such approximate distance preserving subgraphs are called {\em spanners}. Construction of spanners with both additive and multiplicative stretch have been studied extensively~\cite{Erd64, ADDJS93, ACIM99, BTMP05, WOO10, AB16, ABP17}. Fault tolerant version of spanners were first introduced in the geometric setting~\cite{Luk99, LNS02, CZ04}. For $k$ edge failures, construction of a $(2l-1)$ multiplicative spanner of size $\tildeO(kn^{1+1/l})$, for any $k,l\ge 1$, was provided in~\cite{CLPR09} whereas for $k$ vertex failures, upper bound on size is known to be $\tildeO(k^{2-1/l}n^{1+1/l})$~\cite{DK11}. In case of single vertex failure in an undirected graph, construction of a $O(n \log n)$ sized subgraph that preserves distances within a multiplicative factor of $3$, is known due to~\cite{BK13}. The bound on the size was later improved to $3n$~\cite{PP14}. Braunschvig, Chechik and Peleg~\cite{BCPS15} initiated the study of additive spanners. For $\beta$-additive spanner, Parter and Peleg~\cite{PP14} provided a $\Omega(n^{1+\epsilon(\beta)})$ size lower bound where $\epsilon(\beta)\in (0,1)$. They also constructed a $4$-additive spanner of size $O(n^{4/3})$ that is resilient to single edge failure. For single vertex failure, constructions of additive spanners were given in~\cite{Par14, BGGLP15}. Very recently, for any fixed $k \ge 1$, construction of a sub-quadratic size $2$-additive spanner resilient to $k$ edge or vertex failures has been shown for unweighted undirected graphs~\cite{BGPW17}. In the same paper, authors also show that to achieve $O(n^{2-\epsilon})$ upper bound, one must allow $\Omega(\epsilon k)$ additive error.
	
	Designing distance oracle is another important problem and also has been studied in edge failure model. The objective is to build a fault tolerant data structure that can answer queries about the distances in a given graph. For single edge failure the problem was first studied in~\cite{DTCR08}. Construction of $\tildeO(n^2)$-space and $O(1)$-query time oracle is known for single edge failure due to~\cite{BK09}. In case of dual edge failures, near optimal $\tildeO(n^2)$ size and $\tildeO(1)$-query time oracle was given in~\cite{DP09}. The problem has also been studied under the restriction of bounded edge weights~\cite{GW12, WY13}. For general $k$ edge failures, Bil{\`{o}} {\em et al.}~\cite{BGLP16} gave a construction of $O(kn\log^2 n)$ size data structure that can report distance from a single source within multiplicative factor of $(2k+1)$ in time $O(k^2 \log^2 n)$.
	
	Another closely related problem is the replacement path problem where given a source and destination vertex and an edge, the objective is to find a path from source to destination avoiding that particular given edge. Though the problem was initially defined for single edge failure, later it was extended to multiple edge failures also. Readers may refer to~\cite{Williams11, RZ12, GW12, WY13} for recent progresses on this problem.

	\subsection{Our technique}
	Before exhibiting the technique behind our result, we first state a simple observation. If we just store any shortest path tree rooted at $s$, then even after $k$ weight increment that tree will preserve the distance from $s$ to $t$ (for any $t$) within $k$ additive error. It is also necessary to include a shortest path tree inside a $k$-WTSS, otherwise we can never hope to get exact distances even when $k=0$. Now since weight of any path can be increased by at most $k$, after including any shortest $s-t$ path in a $k$-WTSS it is not required to include another $s-t$ path that has weight more than or equal to $dist(s,t)+k$.
	
	We argue that for the construction of $k$-WTSS it suffices to concentrate on any single vertex $t$ and try to build a subgraph such that the distance between $s$ and $t$ is preserved under weight increment and we call such a subgraph a $k$-WTSS($t$). This is because of the application of Locality Lemma (see Section~\ref{sec:localLem}), a variant of which also appears in~\cite{BCR16}. Locality Lemma actually says slightly more, that if we can construct such a subgraph for a particular vertex $t$ with an additional property that in-degree of $t$ in the subgraph is bounded by some value $c$, then we can get a $k$-WTSS of size at most $c n$.
	
	 So from now on we can only talk about constructing $k$-WTSS($t$). Let us take a toy example which provides a motivation behind our technique. Let the input graph be $G$ and $dist(s,t)=d$. Suppose $G$ is such that it can be decomposed into $k+1$ disjoint subgraphs $G_0,\cdots,G_{k}$ where for $0 \le i \le k-1$, $G_i$ contains all the $s-t$ paths of weight $d+i$ present in $G$ and any $s-t$ path in $G_i$ has weight exactly $d+i$. In general such a decomposition may not exist. However, if it exists then it is not hard to get such a decomposition. Now given such a decomposition, we compute a $k$-FTRS($t$) of $G_0$ and for $i \in [k-1]$, a $(k-i-1)$-FTRS($t$) of $G_i$ and then take the union of them. We claim that the obtained subgraph will be a $k$-WTSS($t$). Say after weight increment, the shortest distance between $s$ and $t$ is $d+j$ for $1 \le j < k$. Our assumption on $j$ is justified because $j=0,k$ cases are trivial as we have included $k$-FTRS($t$). For a similar reason we can also assume that all the shortest paths now have weight at least $d+j+1$. Without loss of generality we further assume that no weight increment happens on the edges of the current shortest path. The justification of this assumption is provided later. Due to our assumption on decomposition of $G$, we know that the total increase in weight on the edges of $G_j$ is bounded by $k-(j+1)$ which also implies that at most $k-(j+1)$ many edges of $G_j$ are affected by weight increment. This is because our increment function is integer valued. Note that this is the place where integer valued increment plays a crucial role. However by our construction, we have included $(k-j-1)$-FTRS($t$) of $G_j$ in our subgraph. Thus even if we consider removal of those affected edges then as there is a path in $G_j$ on which there is no weight increment, so by the definition of $(k-j-1)$-FTRS($t$) there will be a surviving path included in our constructed subgraph. This proves the correctness. Also by the result of~\cite{BCR16}, in-degree of $t$ of each $(k-i-1)$-FTRS($t$) of $G_i$ is bounded by $2^{k-i-1}$ and hence total in-degree of $t$ in the constructed $k$-WTSS($t$) is bounded by $2^{k+1}$. Hence we get a $k$-WTSS of size at most $2^{k+1} n$.
	 
	 We have mentioned before that there may not exist the above type of decomposition for an arbitrary graph. In general, if we consider a subgraph by taking all the $s-t$ paths upto some specific weight, then that particular subgraph may also contain a $s-t$ path with larger weight. At this point, argument stated in the last paragraph fails. However, the nice thing is that if we just consider all the shortest paths and build a subgraph then it is true that there will not be any $s-t$ path of larger weight in that subgraph. Now if we use the construction of $k$-FTRS on this shortest path subgraph, then we can guarantee the preservation of distances as long as the distances do not change even after the weight increment. Though if the distance changes, we can not say anything. This is the main challenge that we overcome in our algorithm. For that purpose we use the properties of the farthest min-cut of the shortest path subgraph.
	 
	 Baswana {\em et al.}~\cite{BCR16} used the concept of farthest min-cut to construct $k$-FTRS. In their work, they first computed a series of $k$ farthest min-cuts by taking source sets in some nested fashion. Then they calculated a max-flow from the final source set and kept the incoming edges of $t$ having non-zero flow. We further exploit their technique in this paper to get our algorithm. We consider the shortest path subgraph and then compute a series of farthest min-cuts similar to~\cite{BCR16}. However as mentioned in the last paragraph, in this way we just get $k$-FTRS($t$) of the shortest path subgraph. Now let us take the farthest min-cut considering $s$ as source. Since it is a $(s,t)$-cut of the shortest path subgraph, removal of it destroys all the shortest $s-t$ paths present in the original graph. Now if we again compute the shortest path subgraph, we will get a subgraph containing only $s-t$ paths of weight $d+i$, for some $i>0$. Then we can process this new subgraph as before to compute a sequence of $k$ farthest min-cuts and remove the first one. We proceed in this way until we reach at a point that we are left with $s-t$ paths of weight at least $d+k$.
	 
	 Now let us compare the situation with what we have already discussed with our toy example. Removal of cut edges only helps us to generate some subgraph of each of $G_i$'s. However computing $k$-FTRS($t$) of just some subgraph of $G_i$'s may not be sufficient to get $k$-WTSS($t$). Thus for each $G_i$, we try to get a lot of subgraphs of it so that when we combine $k$-FTRS($t$) of all of them, we get the same advantage that we got from computing $(k-i-1)$-FTRS($t$) of $G_i$ in the toy example. One way of getting a lot of subgraphs of $G_i$ is to try out removal of different cuts (not just the farthest one). Obviously we cannot try for all possible cuts, because there can be too many. Moreover, each time to reach at a subgraph of weight $d+i$ we may have to remove a series of cuts. As a result we may end up with exponentially many choices on removal of cuts to get all possible subgraphs of $G_i$.
	 
	  The good thing is that it suffices to use just a series of $k$ farthest min-cuts computed before for the purpose of removal. This will reduce the number of choices to only $k^i$ for any fixed $G_i$. In our algorithm we establish a slightly better bound on the number of subgraphs of $G_i$ we need to consider to construct a $k$-WTSS($t$). In the proof we use $k$-dimensional vectors to efficiently enumerate all of these subgraphs. After getting those subgraphs we apply a construction similar to that of $k$-FTRS($t$) from~\cite{BCR16} to get a bound on in-degree of $t$. We emphasize that actually we cannot directly apply algorithm of~\cite{BCR16} in a black box fashion on each of the subgraphs of $G_i$ that we consider, because in that case it will not give us the claimed bound.
	 
	 In this paper we consider $k$-WTSS with respect to the weight increment on edges. Instead, it is also possible to take weights on the vertices and performing increment over them. However, one can directly apply our result by splitting each vertex $v$ into two vertices $v_i$ and $v_o$ where all the incoming and outgoing edges of $v$ are respectively directed into $v_i$ and directed out of $v_o$, and then considering an edge $(v_i,v_o)$ with the weight equal to that on the vertex $v$.

	\subsection*{Organization of the paper}
	We discuss useful notations and some already known results about farthest min-cut in Section~\ref{sec:prelim}. Then in Section~\ref{sec:shortFarthest} we provide an algorithm to compute farthest min-cut of the shortest path subgraph and a few important properties about it. Next in Section~\ref{sec:localLem}, we reduce the problem of finding $k$-WTSS to that of finding $k$-WTSS($t$) for some specific vertex $t$ using Locality Lemma. We describe our main algorithm and its correctness in Section~\ref{sec:mainAlgo}. We also present several lower bound results in Section~\ref{sec:sizeLB} and Section~\ref{sec:generalLB}.
	
	\section{Preliminaries}
	\label{sec:prelim}
	\paragraph*{Notations:}
	For any positive integer $r$, we denote the set $\{1,2,\cdots,r\}$ by $[r]$. Throughout this paper we use $\mathbb{N}$ to indicate the set of natural numbers including zero. For any $k$-dimensional vector $\sigma$ and $i \in [k]$, we use the notation $\sigma(i)$ to denote the value of the $i$-th coordinate of $\sigma$. Given a directed graph $G=(V,E)$ on the set of vertices of size $|V|=n$ and set of edges of size $|E|=m$ with a weight function $w$ defined on the set of edges, a source vertex $s \in V$ and a destination vertex $t \in V$, we use the following notations throughout this paper.
	\begin{itemize}
	\item $V(G), E(G):$ the set of vertices and edges of $G$ respectively.
	\item $w(P):$ weight of any path $P$.
	\item $dist_{G,w}(x,y):$ the shortest distance between any two vertices $x$ and $y$ in $G$ when weight of each edge is defined by the weight function $w$.
	\item $G + (u,v):$ the graph obtained by adding an edge $(u,v)$ to the graph $G$.
	\item $G \setminus F:$ the graph obtained by removing the set of edges $F$ from the graph $G$.
	\item $In(A):$ the set of all vertices in $V \setminus A$ having an outgoing edge to some vertex in $A \subseteq V$.
	\item $Out(A):$ the set of all vertices in $V \setminus A$ having an incoming edge from $A \subseteq V$.
	\item $In$-$Edge(A):$ the set of edges incoming to $A \subseteq V$.
	\item $Out$-$Edge(A):$ the set of edges out of $A \subseteq V$.
	\item $P[x,y]:$ the subpath of a path $P$ from a vertex $x$ to $y$.
	\item $P \circ Q:$ the path formed by concatenating paths $P$ and $Q$ assuming the fact that last vertex of $P$ is same as first vertex of $Q$.
	\item $E(f):$ the set of edges $e$ such that under a given flow $f$, $f(e)\ne 0$.
	\item $MaxFlow(G,S,t):$ any maximum valued flow in $G$ from a source set $S$ to $t$.
	\item $G_{short}:$ the shortest path subgraph of $G$, i.e., union of all shortest $s-t$ paths in $G$.
	\item $ShortMaxFlow(G,S,t):$ any maximum valued flow returned by $MaxFlow(G_{short},S,t)$.
	\end{itemize}
	The following definition introduces the notion of $k$-WTSS with respect to a fixed vertex $t$.
	\begin{definition}[$k$-WTSS($t$)]
	Given a graph $G$ with weight function $w$, a source vertex $s \in V(G)$, another vertex $t \in V(G)$ and an integer $k \ge 1$, a subgraph $H_t=(V(G),E')$ where $E' \subseteq E(G)$ is said to be {\em $k$-WTSS($t$)} of $G$ if for any weight increment function $I:E(G) \to \mathbb{N}$ such that $\sum_{e \in E(G)} I(e)\le k$, the following holds: for the weight function defined by $w'(e)=w(e)+I(e)$ for all $e \in E(G)$, $dist_{G,w'}(s,t)=dist_{H_t,w'}(s,t)$.
	\end{definition}
	We can easily extend the above definition for increment function $I:E(G) \to \mathbb{R}$. However in the above definition we just consider $I:E(G) \to \mathbb{N}$ because only with this extra restriction we can hope for sparse $k$-WTSS($t$) (see Section~\ref{sec:generalLB}). Following is an alternative definition of $k$-WTSS in terms of $k$-WTSS($t$).
	\begin{definition}
	A subgraph $H$ is {\em $k$-WTSS} of $G$ if and only if it is $k$-WTSS($t$) for all $t \in V(G)$.
	\end{definition}

	\subsection{Max-flow and farthest min-cut}
	Algorithm described in this paper heavily exploits the connection between min-cut, max-flow and the number of edge disjoint paths present in a graph. Let us start with the following well known fact.
	\begin{theorem}
	In any graph with unit capacity on edges, there is a flow of value $r$ from a source set $S$ to a destination vertex $t$ if and only if there exist $r$ edge disjoint paths that originate from the set $S$ and terminate at $t$.
	\end{theorem}
	Now we define $(S,t)$-min-cut in a graph $G$.
	\begin{definition}[$(S,t)$-min-cut]
	In any graph $G$ an {\em $(S,t)$-cut} is a set of edges $C \subseteq E(G)$ such that every path from any vertex $s \in S$ to $t$ must pass through some edge in $C$. An $(S,t)$-cut is called {\em $(S,t)$-min-cut} if it has the smallest size among all other $(S,t)$-cuts.
	\end{definition}
	Any $(S,t)$-cut $C$ partitions the vertex set $V(G)$ into two subsets $A(C)$ and $B(C)$ where $A(C)$ is the set of all the vertices reachable from $S$ in $G\setminus C$ and $B(C)=V(G)\setminus A(C)$. Note that $S \subseteq A(C)$ and $t \in B(C)$. From now onwards, we assume this pair of vertex sets $(A(C),B(C))$ to be output of a function $Partition(G,C)$.
	
	For our purpose we do not just consider any $(S,t)$-min-cut, instead we consider the farthest one.
	\begin{definition}[Farthest Min Cut]
	Let $S$ be a source set and $t$ be a destination vertex in any graph $G$ and suppose for any $(S,t)$-min-cut $C$, $(A(C),B(C))=Partition(G,C)$. Any $(S,t)$-min-cut $C_{far}$ is called {\em farthest min-cut}, denoted by $FMC(G,S,t)$, if for any other $(S,t)$-min-cut $C$, it holds that $A(C)\subsetneq A(C_{far})$.
	\end{definition}
	
	The following lemma given by Ford and Fulkerson establishes the uniqueness of farthest min-cut and also provides an algorithm to compute it.
	\begin{lemma}{\cite{FF62}}
	\label{lem:uniqueFMC}
	Suppose $f$ be a max-flow in $G$ from any source set $S$ to $t$ and $G_f$ be the corresponding residual graph. If $B$ be the set of vertices from which there is a path to $t$ in $G_f$ and $A=V(G)\setminus B$, then the set $C$ of edges that start at $A$ and terminate at $B$ is the unique farthest $(S,t)$-min-cut.
	\end{lemma}
	Now we state following three important properties of farthest min-cut from~\cite{BCR16}.
	\begin{lemma}{\cite{BCR16}}
	\label{lem:BCR1}
	For any graph $G$, a source set $S$ and a destination vertex $t$, let $C=FMC(G,S,t)$ and $(A,B)=Partition(G,C)$ and for any edge $(s,b) \in (S \times B)$ define $G'=G+(s,b)$. Then the value of max-flow from $S$ to $t$ in $G'$ is exactly one greater than that in $G$ and $FMC(G',S,t)=C \cup \{(s,b)\}$.
	\end{lemma}
	
	\begin{lemma}{\cite{BCR16}}
	\label{lem:BCR2}
	Consider a source vertex $s$ and a destination vertex $t$ in any graph $G$. Let $S \subseteq V(G)$ such that $s \in S$ and $t \not \in S$ and $f$ be a max-flow from $S$ to $t$ in $G$ and $C=FMC(G,S,t)$, $(A,B)=Partition(G,C)$. Then we can always find a max-flow $f_{max}$ from $s$ to $t$ such that $E(f_{max}) \subseteq E(A) \cup E(f)$.
	\end{lemma}
	
	\begin{lemma}{\cite{BCR16}}
	\label{lem:BCR3}
	Consider a source vertex $s$ and a destination vertex $t$ in any graph $G$. Let $S \subseteq V(G)$ such that $s \in S$, $t \not \in S$ and $(A,B)=Partition(G,FMC(G,s,t))$. Then for $(A',B')=Partition(G,FMC(G,S,t))$, $B' \subseteq B$.
	\end{lemma}
	
	\section{Farthest Min-cut of Shortest Path Subgraph}
		\subsection{Computing farthest min-cut of shortest path subgraph}	
		\label{sec:shortFarthest}
	In this section we give an algorithm to find farthest min-cut of the shortest path subgraph of a given graph. We are given a weighted directed graph $G$ with two vertices $s$ and $t$. The weight of each edge of $G$ is defined by a weight function $w:E(G)\to \mathbb{R}$. Let $dist_{G,w}(s,t)=d$. We denote the set of all $s-t$ paths of weight $d$ under the weight function $w$ by $\mathcal{P}_{d}$ and the corresponding underlying subgraph (just the union of all paths in $\mathcal{P}_{d}$) of $G$ by $G_{short}$, more specifically, $V(G_{short})=V(G)$ and $E(G_{short})=\{e\mid e \in P \text{ for some }P \in \mathcal{P}_d\}$.
\begin{definition}
For any graph $G$ and two vertices $s$ and $t$, for any source set $S \subseteq V(G_{short})$, {\em farthest min-cut of shortest path subgraph}, denoted as $FSMC(G,S,t)$ is defined by $FMC(G_{short},S,t)$.
\end{definition}	 
	 For any $(S,t)$-cut $C$ of $G_{short}$ we define the partition function by $ShortPartition(G,C)=Partition(G_{short},C)$.
	 
	 Now given any graph $G$ and two vertices $s$ and $t$, we want to generate the subgraph $G_{short}$. The procedure is as follows:
	 For each edge $(u,v) \in E(G)$ we include $(u,v)$ in a new edge set $E'$ if $dist_{G,w}(s,u)+w(u,v)+dist_{G,w}(v,t)=dist_{G,w}(s,t)$. Then output the graph $G'=(V(G),E')$.
	 
	 It is easy to observe that $G'=G_{short}$ because an edge $e=(u,v) \in P $ for some $P \in \mathcal{P}_d$ if and only if $dist_{G,w}(s,u)+w(u,v)+dist_{G,w}(v,t)=dist_{G,w}(s,t)$.
	 
	 We can implement the above procedure by first storing the values of $dist_{G,w}(s,u)$ for all $u \in V(G)$ and $dist_{G,w}(v,t)$ for all $v \in V(G)$ which can be done in time $O(mn)$~\cite{Bellman58, Ford56} where $|V(G)|=n$ and $|E(G)|=m$. Hence the time complexity to output the subgraph $G'=G_{short}$ will be $O(mn)$.
	 
	 Now we can simply apply well known Ford-Fulkerson algorithm~\cite{FF62} on the subgraph $G_{short}$ to find the $ShortMaxFlow(G,S,t)$ and $FSMC(G,S,t)$. The correctness of $FSMC(G,S,t)$ follows from Lemma~\ref{lem:uniqueFMC} applying on the subgraph $G_{short}$.

	\subsection{Disjoint shortest path lemma}
	Let us choose any $r \in \mathbb{N}$ and then consider the following: Set $S_1=\{s\}$ and for $i \in [r]$, define $C_i=FSMC(G,S_i,t)$, $(A_i,B_i)=ShortPartition(G,C_i)$ and $S_{i+1}=(A_i \cup Out(A_i)) \setminus \{t\}$. Let $E' \subseteq E(G)$ such that $E' = \{(u_1,v_1),\cdots,(u_r,v_r)\}$, where $ (u_i,v_i) \in C_i$. 
	
	Now let us introduce an auxiliary graph $G'=G+(s,v_1)+\cdots +(s,v_r)$ and set $w(s,v_i)=dist_{G,w}(s,v_i)$ for $i \in [r]$. Suppose $f$ be a max-flow from $S_{r+1}$ to $t$ in the shortest path subgraph of $G$ and $\mathcal{E}(t)$ be the set of incoming edges of $t$ having nonzero flow value assigned by $f$. Now consider a new graph $G^*=(G' \setminus In$-$Edge(t))+\mathcal{E}(t)$.
	
	\begin{lemma}
	     \label{lem:FSMCpath}
	     There will be at least $r+1$ disjoint paths in $G^*$ each of weight equal to $dist_{G,w}(s,t)$.
	     \end{lemma}
	     
	     \begin{center}
	     \input{auxiliary}
	     \end{center}
	     Note that a similar claim was shown in~\cite{BCR16}. However, our claim is slightly more general because we consider an edge set $E'$ where the edges belong to $E'$ may not lie on a single $s-t$ path in $G$ and also we comment on the weight of the disjoint paths. Both these requirements are crucial for the proof in Section~\ref{sec:consUB}. Fortunately, the proof in~\cite{BCR16} does not rely on the fact that those edges $(u_i,v_i)$'s are part of a single $s-t$ path. For the sake of completeness we include the proof here.
	     
	     Let us denote the shortest path subgraph (union of all minimum weight $s-t$ paths) of $G$ and $G'$ by $G_{short}$ and $G'_{short}$ respectively. Now let us introduce a series of subgraphs $G_i$'s as follows: 
	     $$G_1=G_{short}\text{, }G_i=G_{short}+(s,v_1)+\cdots+(s,v_{i-1}) \text{ for } 2 \le i \le r+1.$$ Note that $G_{i+1}=G_i + (s,v_i)$. Since $w(s,v_i)=dist_{G,w}(s,v_i)$, the edge $(s,v_i)$ must belong to $G'_{short}$. This is because $(u_i,v_i)$ already lie on some shortest $s-t$ path, say $P$ in $G$ which is a subgraph of $G'$. Then $(s,v_i)\circ P[v_i,t]$ will be a minimum weight $s-t$ path. Hence $G_{r+1}=G'_{short}$. Now let us first prove the following.
	\begin{lemma}
	\label{lem:shortSubgraph}
	For any two vertices $u,v \in V(G_{short})$, any $u-v$ path in $G_{short}$ has weight $dist_{G,w}(u,v)$.
	\end{lemma}
	\begin{proof}
	For the sake of contradiction, suppose $P$ be a $u-v$ path in $G_{short}$ that has weight strictly greater than $dist_{G,w}(v,t)$. By the definition of $G_{short}$, for all $e \in P$, $e$ must lie in some shortest $s-t$ path in $G$. If all the edges in $P$ lies in the same shortest $s-t$ path then the claim is trivial. Otherwise there must exists two consecutive edges $e_1=(x,y)$ and $e_2=(y,z)$ such that they lie in two different shortest $s-t$ paths, say $P_1$ and $P_2$ respectively. Observe that $w(P_1[s,y])=w(P_2[s,y])$, otherwise either of $P_1$ or $P_2$ is not a shortest $s-t$ path. Now consider the path $P_3=P_1[s,y]\circ (y,z) \circ P_2[z,t]$. By construction $w(P_3)=w(P_1[s,y])+ w(y,z) + w(P_2[z,t])=w(P_2)$ as $w(P_1[s,y])=w(P_2[s,y])$. So we get another shortest $s-t$ path such that both $e_1$ and $e_2$ lie in it. We can continue this process to get a new shortest $s-t$ path such that all $e \in P$ lie in it and this concludes the proof.
	\end{proof}
	     
	     \begin{claim}
	     \label{clm:auxCut}
	     $C_i=FMC(G_i,S_i,t)$.
	     \end{claim}
	     \begin{proof}
	     We know that $C_i=FMC(G_{short},S_i,t)$. Observe that for each $j < i$, we add $Out(A_j) \setminus \{t\}$ to $S_{j+1}$. Thus for each edge $(s,v_j)$, $j < i$, both the endpoints lie inside the source set $S_i$. Hence $C_i$ is also equal to the $FMC(G_i,S_i,t)$.
	     \end{proof}
	     \begin{claim}
	     \label{clm:auxFlow}
	     The value of a max-flow from $s$ to $t$ in $G'_{short}$ is at least $r+1$.
	     \end{claim}
	     \begin{proof}
	     We use induction to show that the value of a max-flow from $s$ to $t$ in $G_i$ is at least $i$ for $i \in [r+1]$. The base case $i=1$ is trivial because $G_1=G_{short}$ and there is a $s-t$ path in $G_{short}$. Now for the induction argument, let us first take $(A,B)=ShortPartition(G,FMC(G_i,s,t))$. Then by applying Lemma~\ref{lem:BCR3} on the graph $G_{i}$, we say that $B_i \subseteq B$. Hence by Lemma~\ref{lem:BCR1}, we can argue that the value of a max-flow from $s$ to $t$ in $G_{i+1}$ is at least one more than that in $G_{i}$ which in term is at least $i$ by the induction argument, i.e., the value of a max-flow from $s$ to $t$ in $G_{i+1}$ is at least $i+1$.
	     \end{proof}
	     
	     Now consider a new subgraph $G^*_{short}=(G_{r+1} \setminus In$-$Edge(t))+\mathcal{E}(t)$. Note that $G^*_{short}$ is a subgraph of $G^*$.
	\begin{proof}[Proof of Lemma~\ref{lem:FSMCpath}]
	     As $f$ be a max-flow from $S_{r+1}$ to $t$ in $G_{short}$ and both the endpoints of the edge $(s,v_i)$, for any $i \in [r]$ are inside $S_{r+1}$, the flow $f$ is also a max-flow from $S_{r+1}$ to $t$ in $G_{r+1}=G'_{short}$. Now by applying Lemma~\ref{lem:BCR2} on the graph $G'_{short}$, we can get another max-flow $f_{max}$ from $s$ to $t$ such that $E(f_{max}) \subseteq E(A_{r+1}) \cup E(f)$. As $f_{max}$ terminates $t$ using edges from $\mathcal{E}(t)$, so it is a valid flow also in $G^*_{short}$. So the value of a max-flow from $s$ to $t$ in $G^*_{short}$ is also at least $r+1$. Now as by Lemma~\ref{lem:shortSubgraph} every $s-t$ path in $G^*_{short}$ is of weight equal to $dist_{G,w}(s,t)$ and $G^*_{short}$ is a subgraph of $G^*$, claim of the lemma follows.
	     
	     \end{proof}

	  \section{Construction of $k$-WTSS and Locality Lemma}
	\label{sec:localLem}
	Let us first recall the problem. We are given a graph $G$ along with a weight function $w:E(G) \to \mathbb{Z}$ and a source vertex $s$. Now suppose for every $e \in E(G)$, $w(e)$ is increased by a weight increment function $I:E(G) \to \mathbb{N}$ such that total increase in weight is bounded by $k$, i.e., $\sum_{e \in E(G)}I(e) \le k$ and we denote the new weight function (after increase in weight) by $w'$ where $w'(e)=w(e)+I(e)$. The problem is to find a subgraph $H$ such that for any vertex $t \in V(G)$ in $H$ there always exists an $s-t$ path of weight $dist_{G,w'}(s,t)$. We call this subgraph $H$ a $k$-WTSS. Now if we just want the requirement of existence of path in $H$ to be true for a fixed vertex $t$ instead of all vertices, then we call such a subgraph $k$-WTSS($t$). In this section we reduce the problem of finding $k$-WTSS to the problem of finding $k$-WTSS($t$) for any fixed vertex $t \in V(G)$. The following lemma, a variant of which also appears in~\cite{BCR16}, serves our purpose.
	\begin{lemma}[Locality Lemma]
	\label{lem:localLem}
Let there be an algorithm $\mathcal{A}$ that given a graph $G$ and a vertex $t\in V(G)$, generates a subgraph $H_t$ of $G$ such that:
\begin{itemize}
\item $H_t$ is a $k$-WTSS($t$); and
\item in-degree of $t$ in $H_t$ is bounded by a constant $c_k$.
\end{itemize} 
Then one can generate a $k$-WTSS of $G$ such that it has only $c_k\cdot n$ edges.
\end{lemma}
 \begin{proof}
  
Given the algorithm $\mathcal{A}$ we design another algorithm $\mathcal{A}'$ which generates $k$-WTSS in $n$ rounds. We consider some arbitrary ordering among the vertices as $(v_1, v_2,\cdots,v_n)$. The description of algorithm $\mathcal{A}'$ is as follows:
In the first round we start with the graph $G_0=G$ which is trivially a $k$-WTSS and generate another graph $G_1$ such that in $G_1$ in-degree of $v_1$ is bounded by $c_k$ and $G_1$ is a  $k$-WTSS. Similarly in the $i$-th round, a graph $G_i$ is generated such that in-degree of every vertex $v_j$ for $j\le i$ is bounded by $c_k$ and $G_i$ is $k$-WTSS.
 
 Now we describe round $i$ in details. We start with a graph $G_{i-1}$ which we know is a $k$-WTSS and in-degree of any vertex $v_j$ for $j < i$ is bounded by $c_k$. Let $H_i$ be the $k$-WTSS($v_i$) output by algorithm $\mathcal{A}$. We define $G_i$ to be a subgraph of $G_{i-1}$ where the incoming edges of $v_i$ is restricted to that present in $H_i$. Hence this process assures that in $G_i$ the in-degree of vertices $v_1, \cdots ,v_i$ are bounded by $c_k$.
 
 Now we need to prove that for each $i \in [n]$, $G_i$ is also a $k$-WTSS and we do that inductively. The base case is true as $G_0=G$ is trivially $k$-WTSS. Next assuming $G_{i-1}$ is a $ k$-WTSS we prove the same for $G_i$. Now consider any increment function $I$. Let $F$ be the set of edges for which $I$ has non zero value in $G_{i-1}$ i.e., $F=\{ e\in E(G_{i-1}) \mid I(e)>0 \}$. Suppose the new weight function is $w'$ defined by $w'(e)=w(e)+I(e)$. Now consider any vertex $t$. Let $dist_{G,w'}(s,t)=d'$ and hence by the induction argument, $dist_{G_{i-1},w'}(s,t)=d'$. Suppose the corresponding path is $P$ in $G_{i-1}$. We need to show there exist an $s-t$ path $R$ of weight $d'$ in $G_i$ such that $w'(R)=d'$. If path $P$ does not pass through $v_i$ set $R=P$. Otherwise consider the segments $P[s,v_i]$ and $P[v_i,t]$. We have that $w'(P[s,v_i])+w'(P[v_i,t])=d'$. Observe that $w'(P[s,v_i])=dist_{G_{i-1},w'}(s,v_i)$ otherwise $P$ cannot be a shortest $s-t$ path under $w'$. Now as $G_i$ differs from $G_{i-1}$ only at the incoming edges of $v_i$, path segment $P[v_i,t]$ remains intact. As $H_i$ is a $k$-WTSS($v_i$) for $G_{i-1}$, there exist an $s-v_i$ path, say $R'$ in $H_i$ of weight $dist_{G_{i-1},w'}(s,v_i)$. By the construction $G_i$ contains $H_i$ and hence $R'\circ P[v_i,t]$ is a walk from $s$ to $t$ of weight $w'(R')+w'(P[v_i,t])=dist_{G_{i-1},w'}(s,v_i)+w'(P[v_i,t])=d'$ in $G_i$. Removing loops we get our desired path $R$ of weight at most $d'$. Now as $G_i$ is a subgraph of $G_{i-1}$, we can conclude that $w'(R)=d'$. Hence $G_i$ is a $k$-WTSS.
\end{proof}

	 	\section{Construction of $k$-WTSS($t$)}
	\label{sec:consUB}
	In this section we provide an algorithm to compute a $k$-WTSS($t$) for any fixed vertex $t \in V(G)$ where source vertex is $s$. Without loss of generality let us first assume the following.
	\begin{assumption}
	\label{asm:degBound}
	The out degree of source vertex $s$ is $1$ and the out degree of all other vertices is bounded by $2$.
	\end{assumption}
	For any graph $G$ if $|Out(s)|>1$ then to satisfy our previous assumption, we can simply add a new vertex $s_0$ and add an edge $(s_0,s)$ and set $w(s_0,s)=0$. Then make this new vertex $s_0$ as our new source. For the justification on the bound on out degree of other vertices, we refer the readers to Appendix~\ref{app:assumption}.
	
	\subsection{Description of the algorithm}
	\label{sec:mainAlgo}
		Before describing the algorithm let us introduce some notations that we will use later heavily. Consider any $k$-dimensional vector $\sigma \in \{-1,0,1,\cdots,k\}^{k}$ such that if $\sigma(i)=-1$ then for all $i' >i$, $\sigma(i')=-1$ where $i,i' \in [k]$ and if $\sigma(i) \ne -1$ then for all $i' < i$, $\sigma(i') \ne -1$. We use these vectors to efficiently enumerate all the subgraphs of $G$ for which we want to calculate farthest min-cuts. Now for any such vector $\sigma$ and $r \in [k]$, we recursively define the subgraph $G_{\sigma}$, set of source vertices $S_{\sigma,r}$ and edge set $C_{\sigma,r}$ as follows: if $\sigma=(-1,-1,\cdots,-1)$, $G_{\sigma}$ is the union of all $s-t$ paths in $G$, starting with $S_{\sigma,1}=\{s\}$, for any $r \in [k]$ define $C_{\sigma,r}=FSMC(G_{\sigma},S_{\sigma,r},t)$, $S_{\sigma,r+1}=(A \cup Out(A))\setminus \{t\}$ where $(A,B)=ShortPartition(G_\sigma, C_{\sigma,r})$. For $\sigma \ne (-1,\cdots,-1)$, $G_{\sigma}$ is the union of all $s-t$ paths in $G_{\sigma'}\setminus  C_{\sigma',\sigma(i)+1}$, where $i=\max\{i'\mid\sigma(i')\ne -1\}$ and
	$$\sigma'(i')=\begin{cases}\sigma(i')\quad&\text{if $i' < i$}\\
	-1&\text{otherwise}\end{cases}$$
	Now starting with $S_{\sigma,1}=\{s\}$, if there exists a $s-t$ path of weight $d+i$ then for any $r \in [k]$ define $C_{\sigma,r}=FSMC(G_{\sigma},S_{\sigma,r},t)$, $S_{\sigma,r+1}=(A \cup Out(A))\setminus \{t\}$ where $(A,B)=ShortPartition(G_\sigma, C_{\sigma,r})$; else set $C_{\sigma,r}=\phi$. We refer the reader to Figure~\ref{fig:example} for the better understanding about the graph $G_{\sigma}$. 
	
	\begin{center}
	\input{example}
	\end{center}

	 We are given a weighted directed graph $G$ with weight function $w$ and a source vertex $s$ and a destination vertex $t$. The weight of each edge of $G$ is defined by the weight function $w:E(G)\to \mathbb{Z}$. Overall our algorithm (Algorithm~\ref{alg:main}) performs the following tasks: For different values of $\sigma \in \{-1,0,\cdots,k\}^k$ it computes the sets $C_{\sigma,i}$ and $S_{\sigma,i}$ for $i \in [k]$. Then for each such $\sigma$, it computes max-flow in the shortest path subgraph of $G_\sigma$ by considering $S_{\sigma,k}$ as source and add the edges incident on $t$ with non-zero flow to a set $\mathcal{E}(t)$. At the end, our algorithm returns the subgraph $H_t=(G \setminus In$-$Edge(t))+\mathcal{E}(t)$.
			
			Our algorithm performs the above tasks in the recursive fashion. Starting with $\sigma=(-1,\cdots,-1)$, it first considers the shortest path subgraph of $G_{\sigma}=G$ and performs $k$ iterations on it. At each iteration it computes the farthest min-cut $C_{\sigma,i}$ by considering $S_{\sigma,i}$ as source and $t$ as sink starting with $S_{\sigma,1}=\{s\}$. Then it updates the graph by removing the edges present in $C_{\sigma,i}$ and passes this new graph in the next recursive call. Before the recursive call it also updates the $\sigma$ by incrementing the value of $\sigma(j)$ by one and passes the updated value of $\sigma$ to the recursive call. Here $j$ is a parameter which denotes that the smallest coordinate of $\sigma$ that has value $-1$. Initially $j$ was set to $1$ and before the next recursive call we increment its value by one. At the end of each iteration our algorithm updates the source set to $S_{\sigma,i+1}$ by including end points of all the edges present in the cut $C_{\sigma,i}$ in the set $S_{\sigma,i}$. At the end of $k$ iterations, the algorithm computes max-flow in the shortest path subgraph of $G_\sigma$ by considering $S_{\sigma,k}$ as source and add the edges incident on $t$ with non-zero flow to a set $\mathcal{E}(t)$.

	     \begin{algorithm}[h]
		\Input{A graph $G$ with weight function $w$ and two vertices $s$ and $t$}
		\Output{A subgraph $H_t$}
		 \tcp{Initialization:} 
		For all $\sigma \in \{-1,0,\cdots,k\}^k$ and $r \in [k]$, set $C_{\sigma,r}$ to be $\phi$;
				 
		 Set $\sigma_{curr}=(-1,\cdots,-1)$;
		   
		   $RecursiveWTSS(G,\sigma_{curr},1)$;
		   
		   Return $H_t = (G \setminus In$-$Edge(t)) + \mathcal{E}(t)$;
		   
		 \caption{Algorithm for computing $k$-WTSS($t$)}
		 \label{alg:main}
		 	\end{algorithm} 	
			
			\begin{procedure}[h]
	\hrule height .7pt\vspace{1mm}
	\TitleOfAlgo{$RecursiveWTSS(G_{curr},\sigma,j)$}\label{alg:recursive}
	\hrule\vspace{1mm}

				\If{there exists an $i \in [j-1]$ such that $\sigma(i) \ge k-j+i-1$}
				{
				
				 \Return;
				 
				}
				
				Define $\sigma_{curr}$ by setting $\sigma_{curr}(j)=0$ and $\sigma_{curr}(i)=\sigma(i)$ for all $i \ne j$;
				
               \If{$dist_{G_{curr},w}(s,t)=d+j-1$}
               {
               	$S_1 \gets \{s\}$;
               	
               	\For{$i=1,\cdots,k$}
               	{
               		$C_{\sigma,i} \gets FSMC(G_{curr},S_{i},t)$;
               		
               		$RecursiveWTSS((G_{curr}\setminus C_{\sigma,i}),\sigma_{curr},j+1)$;
               		
               		$\sigma_{curr}(j) \gets \sigma_{curr}(j)+1$;
               		
               		$(A_i,B_i) \gets ShortPartition(G_{curr},C_{\sigma,i})$;
               		
               		$S_{i+1} \gets (A_i \cup Out(A_i)) \setminus \{t\}$;

               	}

               $f \gets ShortMaxFlow(G_{curr},S_{\sigma,k+1},t)$;
               
               Add incoming edges of $t$ present in $E(f)$ in $\mathcal{E}(t)$;
               	
              }
              \Else
              {
              Define $\sigma_{curr}$ by setting $\sigma_{curr}(j)=0$ and $\sigma_{curr}(i)=\sigma(i)$ for all $i \ne j$;
              
              		$RecursiveWTSS(G_{curr},\sigma_{curr},j+1)$;
              }
              
              \hrule\vspace{1mm}
\end{procedure}

	\subsection{Correctness proof}
	\label{sec:proof}
	 Let us start with the following simple observation.
	 \begin{observation}
	 \label{obs:shortestPath}
	 For any $\sigma \in \{-1,0,\cdots, k\}^k$, any $s-t$ path in $G_{\sigma}$ must have weight at least $d+i-1$ where $i=\min\{i'\mid\sigma(i')= -1\}$.
	 \end{observation}
	 \begin{proof}
	 Now for any $\sigma$, let us consider a sequence of vectors $\alpha_1,\cdots, \alpha_i \in \{-1,0,\cdots, k\}^k$ where $i=\min\{i'\mid\sigma(i')= -1\}$ as follows: for any $1 \le j \le i$,
	$$\alpha_j(i')=\begin{cases}\sigma(i')\quad&\text{if $i' < j$}\\
	-1&\text{otherwise}\end{cases}$$ 
	Note that $\alpha_i=\sigma$. Now we use induction on $j$ to show that any $s-t$ path in $G_{\alpha_j}$ must have weight at least $d+j-1$ and that will prove our claim.
	
	As a base case when $j=1$, as $\alpha_1=(-1,\cdots,-1)$, $G_{\alpha_j}=G$ and hence the claim is trivially true. Now suppose the claim is true for any $j \in [i-1]$ and we need to prove it for $j+1$. By the definition, $G_{\alpha_{j+1}}$ is a subgraph of $G_{\alpha_j}$. By induction hypothesis all the $s-t$ paths in $G_{\alpha_j}$ have weight at least $d+j-1$. If there is no $s-t$ path of weight $d+j-1$ in $G_{\alpha_j}$ then we are done because of our integer valued weight function. Otherwise any such path of weight $d+j-1$ must pass through the cut set $C_{\alpha_j,\sigma(j)+1}$. Now by definition, $G_{\alpha_{j+1}}$ is build by removing the edge set $C_{\alpha_j,\sigma(j)+1}$ from the graph $G_{\alpha_j}$. Hence there will be no $s-t$ path of weight $d+j-1$ in $G_{\alpha_{j+1}}$. Since our weight function is integer valued, the claim follows.
	 \end{proof}
	 
	 Note that the above observation is true only because we consider the range of our weight function $w$ to be $\mathbb{Z}$. Otherwise above observation will trivially be false. 

	Now let us consider any increment function $I:E(G)\to \mathbb{N}$ such that $\sum_{e \in E(G)} I(e) \le k$ and then denote the set of edges with non-zero value of the function $I$ by $F$, i.e., $F=\{e \in E(G)|I(e) > 0\}$. So clearly $|F| \le k$. Now suppose $dist_{G,w'}(s,t)=d'=d+j$ for some $0 \le j \le k$ where $w'(e)=w(e)+I(e)$. Thus we need to show that there also exists an $s-t$ path of weight $d'$ in the subgraph $H_t$ under the new weight function $w'$.
	
	Suppose $P$ be an $s-t$ path in $G$ such that $w'(P)=d'=d+j$. For simplicity let us assume the following.
	\begin{assumption}
	\label{asm:weightP}
	For all $e \in P$, $I(e)=0$.
	\end{assumption}
	In other words we are assuming that $w'(P)=w(P)$. At the end of the current subsection we discuss how to remove this assumption. 
	\begin{lemma}
	\label{lem:path}
	One of the following three cases must satisfy.
	\begin{enumerate}
	\item There exists a $\sigma$ such that $P$ belongs to the subgraph $G_{\sigma}$ where $\sigma(j)=-1$ and for some positive integer $r$, the last edge of $P$ belongs to the edge set $C_{\sigma,r}$.
	\item There exists a $\sigma$ such that $P$ belongs to the subgraph $G_{\sigma}$ where $\sigma(j+1)=-1$, $\sigma(j)\ne -1$ and there is no $i \in [j-1]$ such that $\sigma(i)\ge k-j+i-1$.
	\item There exists a $\sigma$ such that $P$ belongs to the subgraph $G_{\sigma}$ where if $i=\min\{i' \mid \sigma(i') = -1\}$ then $i \le j$ and for any $i' \le i$, $\sigma(i') < k-j+i'-1$ and $P$ passes through all the cut sets $C_{\sigma,1},\cdots, C_{\sigma,k-j+i-1}$.
	\end{enumerate}
	\end{lemma}
	 \begin{proof}
	Here we describe a procedure to find desired $\sigma$ for the path $P$. Let us initialize $\sigma=(-1,\cdots,-1)$. So $G_\sigma=G$ and thus trivially $P$ belongs to $G_\sigma$. Suppose $P$ pass through edges of the cut sets $C_{\sigma,1}, \cdots, C_{\sigma,r_1}$, but does not pass through any edge of $C_{\sigma,r_1+1}$. Note that $r_1$ will be equal to $0$ if $P$ does not pass through any of $C_{\sigma,1}$. Update $\sigma$ by setting $\sigma(1)=r_1$. By the definition of $G_{\sigma}$, $P$ belongs to it. Now suppose $P$ passes through edges of the cut sets $C_{\sigma,1}, \cdots, C_{\sigma,r_2}$, but does not pass through any edge of $C_{\sigma,r_2+1}$. Then update $\sigma$ by setting $\sigma(2)=r_2$. Now proceed in this way until $\sigma(j)$ is set or we reach at a point where for some $i \in [j-1]$, $P$ passes through all the cut sets $C_{\sigma,1},\cdots, C_{\sigma,k-j+i-1}$. This process may stop prematurely if $P$ reaches $t$ before satisfying either of above two conditions, but in that case the last edge, say $(v,t)$ of $P$ must belong to some cut set $C_{\sigma,r}$. Hence we will be in the first case and this completes the proof.
	\end{proof}
	
	Now let us call the path $P$ is of type-$1$, type-$2$ and type-$3$ respectively depending on which of the above three cases it satisfies.

	\paragraph*{Type-$1$:}
	This case is the simplest among the three.
	\begin{lemma}
	\label{lem:correct1}
	If $P$ is a type-$1$ path then $P$ is contained in the subgraph $H_t$.
	\end{lemma}
	\begin{proof}
	Suppose $(v,t)$ is the last edge of the path $P$. Now as $(v,t) \in C_{\sigma,r}$ for some $\sigma$ and $r$, $(v,t) \in \mathcal{E}(t)$. Thus by the construction of the subgraph $H_t$, the edge $(v,t)$ belongs to $H_t$. Also by the construction of the subgraph $H_t$, for all the vertices $u \ne t$, $In$-$Edge(u)$ belong to $H_t$. Hence $P$ must lie completely inside $H_t$.
	\end{proof}
	
		\paragraph*{Type-$2$:}
	As path $P$ belongs to the subgraph $G_{\sigma}$ where $\sigma(j)\ne -1$ and $\sigma(j+1)=-1$, by Observation~\ref{obs:shortestPath}, $w(P) \ge d+j$. However by our Assumption~\ref{asm:weightP}, $w(P)=d+j$ and so it must pass through an edge $(u_r,v_r) \in C_{\sigma,r}$ for all $r \in [k]$. Now consider an auxiliary graph $G'_{\sigma}=G_{\sigma} + (s,v_1) + \cdots + (s,v_{k})$ and extend the weight function $w$ as $w(s,v_r)=w(P[s,v_r])$. Then define another graph $G^*_\sigma=(G'_\sigma \setminus In$-$Edge(t)) + \mathcal{E}(t)$. By Lemma~\ref{lem:FSMCpath} we can claim the following.
	\begin{corollary}
	\label{cor:disjointCur}
	There will be $k+1$ edge disjoint paths in $G^*_\sigma$ each of weight $w(P)$ under weight function $w$.
	\end{corollary}
	
	Now we use the above corollary to conclude the following. The argument is similar to that used in~\cite{BCR16}.
	\begin{lemma}
	\label{lem:correct2}
	If $P$ is a type-$2$ path then there exists an $s-t$ path of weight $d'$ in the subgraph $H_t$ under the new weight function $w'$.
	\end{lemma}
	\begin{proof}
	By Corollary~\ref{cor:disjointCur}, we get $k+1$ edge disjoint paths $P_1,\cdots, P_{k+1}$ each of weight $w(P)=d+j$. Since $|F| \le k$ where $F=\{e \in E(G)|I(e) > 0\}$, at least one of the $k+1$ many edge disjoint paths, say $P_1$ must  survive in $G^*_{\sigma}\setminus F$. If $P_1$ also belongs to the subgraph $H_t$ then we are done. Otherwise $P_1$ must take some of the $(s,v_r)$'s as the first edge and the remaining portion $P_1[v_r,t]$ lies inside $H_t$. Now consider the following path $R=P[s,v_r]\circ P_1[v_r,t]$. By the construction of $G'_{\sigma}$, $w'(R)=w'(P_1)=w(P)$ and this completes the proof.
	\end{proof}

	\paragraph*{Type-$3$:}
	 Suppose $P$ is a type-$3$ path and thus belongs to $G_{\sigma}$ for some $\sigma$ where if $i=\min\{i' \mid \sigma(i') = -1\}$ then $i \le j$ and for any $i' \le i$, $\sigma(i') < k-j+i'-1$ and $P$ passes through all the cut sets $C_{\sigma,1},\cdots, C_{\sigma,k-j+i-1}$. $P$ passes through an edge $(u_r,v_r) \in C_{\sigma,r}$ for all $r \in [k-j+i-1]$. For ease of representation let us define $v_0=s$. Now if there exists a positive integer $r \in [k-j+i-1]$ such that $w(P[v_{r-1},u_{r}]) > dist_{G_\sigma,w}(v_{r-1},u_{r})$, replace the portions of path $P[v_{r-1},u_{r}]$ by the $v_{r-1}-u_{r}$ path of weight  $dist_{G_\sigma,w}(v_{r-1}r,u_{r})$. We do this until there is no such $r$ and after that we call this new path as $P'$.
	\comment{
	\begin{claim}
	\label{clm:wtP}
	$d+i-1 \le w(P') \le d+j$.
	\end{claim}
	\begin{proof}
	From the construction of $P'$ it can be noted that we replace some portion of the original path $P[v,u]$ by a path of smaller weight only and hence total weight can only be decreased.
	
	As we argued before $P'$ lies inside the subgraph $G_{\sigma}$ and $i=\min\{i' \mid \sigma(i')= -1\}$. Hence the lower bound on $w(P')$ follows from Observation~\ref{obs:shortestPath}.
	\end{proof}
	}
	
	Now consider an auxiliary graph $G'_{\sigma}=G_{\sigma} + (s,v_1) + \cdots + (s,v_{k-j+i-1})$ and extend the weight function $w$ as $w(s,v_r)=w(P[s,v_r])$. Next define another graph $G^*_{\sigma}=(G'_{\sigma} \setminus In$-$Edge(t)) + \mathcal{E}(t)$. Now we use a slightly different argument than that used previously.
	
	Let us now analyze by considering the following two cases separately.
	
	\subparagraph*{Case $1$:}  [$w(P[s,u_{1}]) = dist_{G_\sigma,w}(s,u_{1})$]
	
	\begin{claim}
	\label{clm:nonPathEdge}
	There will be at least $k-j+i$ edge disjoint paths in $G^*_{\sigma}$ each of weight at most $d+j$ under weight function $w$. Moreover, at least one path among them will be of weight $d+i-1$.
	\end{claim}
	\begin{proof}
	Let us consider a new weight function $w_1$ as follows:
	$$w_1(e)=\begin{cases}dist_{G_\sigma,w}(s,v_r)\quad&\text{if $e = (s,v_r)$ for some $r \in [k-j+i-1]$}\\
	w(e)&\text{otherwise}\end{cases}$$
	Then by Lemma~\ref{lem:FSMCpath}, $G^*_{\sigma}$ has $(k-j+i)$ edge disjoint paths, say $P_1,\cdots, P_{k-j+i}$ each of weight $d+i-1$ under the new weight function $w_1$ where the weight bound follows from Observation~\ref{obs:shortestPath}. Now since by Assumption~\ref{asm:degBound} the out degree of $s$ is $1$, so $|C_{\sigma,1}|=1$. As both $P$ and $P'$ pass through the edge $(u_1,v_1)$ and $w(P[s,u_{1}]) = dist_{G_\sigma,w}(s,u_{1})$, so from the construction of $P'$ it can be observed that $w_1(s,v_1)=w(s,v_1)$. Now consider the path that takes $(s,v_1)$ as the first edge and say it is $P_1$. Then $$w(P_1)=w(s,v_1)+w(P_1[v_1,t])=w_1(s,v_1)+w_1(P_1[v_1,t])=w_1(P_1)=d+i-1.$$
	For any other path, say $P_2$, clearly $w(P_2)\le w_1(P_2)+(j-i+1)=d+j$ because for any $2 \le r\le k-j+i-1$, $w(s,v_r)-w_1(s,v_r) \le j-i+1$.
	\end{proof}
	So we get $k-j+i$ edge disjoint paths $P_1,\cdots, P_{k-j+i}$ each of weight at most $d+j$ and suppose $P_1$ has weight $d+i-1$. Let us also extend the weight function $w'$ by setting $w'(s,v_r)=w(s,v_r)$ and extend $I$ by setting $I(s,v_r)=0$. If any one of these $k-j+i$ edge disjoint paths, say $Q$ satisfies that $w'(Q)\le d+j$, then we are done. This is because in that case either $Q$ lies inside $H_t$ which makes $Q$ to be our desired path or for some $r \in [k-j+i-1]$, $Q$ must take $(s,v_r)$'s as the first edge and the remaining portion $Q[v_r,t]$ lies inside $H_t$. In the second case, we consider the path $R=P[s,v_r]\circ Q[v_r,t]$. Note that $w'(R)=w'(Q)\le d+j$.
	
	Now we argue that there must exists one path among $k-j+i$ edge disjoint paths such that it will have weight at most $d+j$ under the weight function $w'$. Otherwise for all $r\in [k-j+i]$, $w'(P_r)\ge d+j+1$. Hence $I(P_1) \ge j-i+2$ and $I(P_r) \ge 1$ for all $2 \le r\le k-j+i$. Thus
	
	$$
	\sum_{e \in E(G^*_{\sigma})}I(e)  \ge (j-i+2) + (k-j+i-1) \ge (k+1).
	$$
	
	However as $I(s,v_1)=I(s,v_2)=\cdots =I(s,v_{\sigma(i)})=0$,
	$$\sum_{e \in E(G^*_{\sigma})} I(e) \le \sum_{e \in E(G)} I(e) \le k$$
	which leads to a contradiction.
	
	\subparagraph*{Case $2$:}  [$w(P[s,u_{1}]) > dist_{G_\sigma,w}(s,u_{1})$]
	
	In this case also by the argument used in the first part of the proof of Claim~\ref{clm:nonPathEdge}, we can claim the following.
	\begin{claim}
	\label{clm:nonPathEdge2}
	There will be at least $k-j+i$ edge disjoint paths in $G^*_{\sigma}$ each of weight at most $d+j$ under weight function $w$.
	\end{claim}
	Note that on the contrary to Claim~\ref{clm:nonPathEdge}, now we do not have the extra guarantee that at least of the edge disjoint paths must have weight $d+i-1$. Now just like the previous case, we only need to argue that there must exists one path among $k-j+i$ edge disjoint paths, say $P_1,\cdots, P_{k-j+i}$ such that it will have weight at most $d+j$ under the weight function $w'$ and we will be done.
	
	Now suppose $w(P[s,u_{1}]) = dist_{G_\sigma,w}(s,u_{1}) + l$, for $l > 0$. Consider the path that takes $(s,v_1)$ as the first edge and say it is $P_1$. Then by the argument used in the proof of Claim~\ref{clm:nonPathEdge}, one can show that
	$$w(P_1)=w(s,v_1)+w(P_1[v_1,t])=(dist_{G_\sigma,w}(s,v_{1}) + l)+w_1(P_1[v_1,t])=d+i+l-1.$$
	Let $Q$ be a shortest $s-v$ path in $G_{\sigma}$ under weight $w$. Now since $w(P[s,u_{1}]) > dist_{G_\sigma,w}(s,u_{1})$ and $P$ is a shortest $s-t$ path under the weight $w'$ (recall that $w'=w+I$), $I(Q) \ge l+1$. Moreover, 
	$$\sum_{e \in Q\text{ and }e \not \in P}I(e) \ge l+1.$$
	Now if for all $r\in [k-j+i]$, $w'(P_r)\ge d+j+1$, it must satisfy that $I(P_1) \ge j-i-l+2$ and $I(P_r) \ge 1$ for all $2 \le r\le k-j+i$. 
	
	$$
	\sum_{e \in E(G^*_{\sigma})}I(e)  \ge (l+1) + (j-i-l+2) + (k-j+i-2) \ge (k+1).
	$$
	However as $I(s,v_1)=I(s,v_2)=\cdots =I(s,v_{\sigma(i)})=0$,
	$$\sum_{e \in E(G^*_{\sigma})} I(e) \le \sum_{e \in E(G)} I(e) \le k$$
	which again leads to a contradiction.

	Now from the above we can conclude the following.
	\begin{lemma}
	\label{lem:correct3}
	If $P$ is a type-$3$ path then there exists an $s-t$ path of weight $w(P)=d'$ in the subgraph $H_t$ under the new weight function $w'$.
	\end{lemma}
	
	\paragraph*{Removing Assumption~\ref{asm:weightP}:}
	Suppose $P$ be one of the shortest paths from $s$ to $t$ in $G$ under the new weight $w'$, i.e., $w'(P)=d'$. Then consider the following set $S=\{e \in P | I(e) > 0\}$ and suppose $\sum_{e\in S}I(e) = k' \le k$. Now define the following new weight function:
	$$w''(e)=\begin{cases}w(e)\quad&\text{if $e \in P$}\\
	w'(e)&\text{otherwise}\end{cases}$$
	Now $w''(P)=d'-k'$. Then use the argument same as before to show that there exists a path, say $R$ in $H_t$ of weight at most $d'-k'$ under this new weight function $w''$. Clearly, $w'(R)\le w''(R)+k'=d'$.

	\subsection{Bound on size of $\mathcal{E}(t)$}
	\label{sec:sizebound}
	Before establishing the upper bound on the size of the set of edges $\mathcal{E}(t)$, let us define $C_{\sigma,k+1}=FSMC(G_{\sigma},S_{\sigma,k+1},t)$ for any $\sigma \in \{-1,0,\cdots ,k\}^k$. Now as $FSMC(G_{\sigma},S_{\sigma,i+1},t)=FMC(G^{short}_{\sigma},S_{\sigma,i+1},t)$ for any $i \in [k]$ where $G^{short}_{\sigma}$ is the shortest path subgraph of $G_{\sigma}$, so we can restate Lemma 6.6 from~\cite{BCR16} in the following form.
	\begin{lemma}
	\label{lem:sizeCut}
	For any $i \in [k]$, $|C_{\sigma,i+1}| \le 2 \times |C_{\sigma,i}|$.
	\end{lemma}
	Reader may note that the proof of the above lemma in~\cite{BCR16} crucially relies on our Assumption~\ref{asm:degBound}.
	\begin{lemma}
	\label{lem:sizeBound}
	$|\mathcal{E}(t)| \le e(k-1)!2^k$.
	\end{lemma}
	\begin{proof}
	In our algorithm for each $\sigma \in \{-1,0,\cdots,k\}^k$ we compute the cut sets $C_{\sigma,1},\cdots,C_{\sigma,k}$ and add $|C_{\sigma,k+1}|$ many edges in the set $\mathcal{E}(t)$ if for any $i' \le i$, $\sigma(i') < k-i+i'-1$ where $i =\min \{j \mid \sigma(j) =-1\}$; otherwise we do not compute anything. So the total number of $\sigma$ for which we add edges in $\mathcal{E}(t)$ is bounded by
	$$1+(k-1)+(k-1)(k-2)+\cdots + (k-1)!=(k-1)![1/0!+1/1!+\cdots + 1/(k-1)!] \le e \cdot (k-1)!.$$
	Now by applying Lemma~\ref{lem:sizeCut}, we get that for each such $\sigma$, $|C_{\sigma,k+1}|\le 2^k$ and this proves the claimed bound.
	\end{proof}
	
	\subsection{Complexity analysis}
	\label{sec:complexity}
	Now we analyze the running time of our algorithm to find $k$-WTSS($t$) for some $t \in V(G)$. We first preprocess the input graph to generate a new graph in a way so that the new graph satisfies Assumption~\ref{asm:degBound}. This takes $O(m)$ time (see Appendix~\ref{app:assumption}). Next we apply Algorithm~\ref{alg:main} on this new graph which has $O(m)$ many vertices and edges. By the argument in the proof of Lemma~\ref{lem:sizeBound} we see that our algorithm computes $k$ farthest min-cuts on shortest path subgraphs of $G_{\sigma}$ for $e(k-1)!$ many different $\sigma$'s. Now from the discussion in Section~\ref{sec:shortFarthest}, assuming we have $G_{\sigma}$ explicitly, to generate each such shortest path subgraph on this new transformed graph we need $O(m^2)$ time and then to compute $k$ farthest min-cuts takes total $O(\sum_{i=1}^k |C_{\sigma,i}| \times m)=O(2^km)$ time (see~\cite{FF62}). Finally, one can get $k$-WTSS($t$) of the original graph from that of the transformed graph in $O(m)$ time (see Appendix~\ref{app:assumption}). So overall time needed to compute $k$-WTSS($t$) of any given graph with $n$ vertices and $m$ edges is $O((k-1)!2^{k} m^2)=O((k)^k m^2)$ (by Stirling's approximation). Now since by the Locality Lemma (Lemma~\ref{lem:localLem}), finding $k$-WTSS requires $n$ rounds where in each round we find $k$-WTSS($v$) for some $v \in V(G)$, computing $k$-WTSS takes total $O((k)^km^2n)$ time.

	\section{Lower Bound on the Size of $k$-WTSS}
	\label{sec:sizeLB}
	In this section we give construction of a graph that will establish a lower bound on the size of a $k$-WTSS as stated in Theorem~\ref{thm:LB}. Let us first recall Theorem~\ref{thm:LB}.
	\begin{theorem}
	For any positive integer $k \ge 2$, there exists an positive integer $n'$ such that for all $n > n'$, there exists a directed graph $G$ with $n$ vertices and a weight function $w:E(G) \to \mathbb{Z}$, such that its $k$-WTSS must contain $c \cdot 2^k n$ many edges for some constant $c \ge 5/4$.
	\end{theorem}
	\begin{proof}
	Let us consider $l$ many full binary trees $T_i$ such that for each $1 \le i \le l$, $T_i$ has height $h_i=k-\sum_{j=2}^{i}j$ with root $r_i$. Let $L_i$ be the set of leaves of the tree $T_i$ and thus $|L_i|=2^{h_i}$. Next consider $L=\cup_i L_i$ and another set of $n$ vertices $X$. Finally define a graph $G$ with $V(G)=\{s\} \cup (\cup_i V(T_i)) \cup X$ and $E(G)=\{(s,r_i) | 1 \le i \le l\} \cup \{(u,v)| u \in L, v \in X\} \cup (\cup_i E(T_i))$. Choose largest $l$ such that $\sum_{j=2}^{l} j \le k$. Now let us consider the following weight function $w:E(G) \to \mathbb{N}$,
	$$
	w(e) = \begin{cases}\sum_{j=2}^{i}j + i \quad & \text{if $e=(s,r_i)$}\\
	1 &\text{otherwise}  \end{cases}
	$$
	Clearly, $ |E(G)| = l + \sum_{i=1}^{l}(2|L_i| - 1) + |L| \times |X| =c \cdot 2^k n $ for some constant $c \ge 5/4$.
	
	It only remains to show that any $k$-WTSS of $G$ must contain all the edges of $G$. Take any vertex $t \in X$ and consider any path $P$ from $s$ to $t$. Suppose the first edge of $P$ is $(s,r_i)$ for some $i$. Then consider the set $S=\{(u,v)| (u,v) \in T_i,\;u \in P \text{ but } v \not \in P \}$. Let us now consider the following increment function $I:E(G) \to \{0, \cdots ,k\}$,
	$$I(e)=\begin{cases}i + 1 - j \quad &\text{if $e=(s,r_j)$ and $1\le j < i$}\\
	1&\text{$e \in S$}\\
	0&\text{otherwise} \end{cases}$$
	Clearly, $\sum_{e \in E(G)}I(e) \le k$ due to the choice of $l$. Also $P$ will be the only shortest path from $s$ to $t$ in $G$ under the new weight function $w'(e)=w(e)+I(e)$, $\forall_{e\in E(G)}$ because all the $s-t$ paths whose first edge is $(s,r_j)$ for $j < i$ and all the $s-t$ paths except $P$, whose first edge is $(s,r_i)$ will now have weight $k + i +1$. This completes the proof.
	\end{proof}

	\section{Lower Bound for More General Model}
	\label{sec:generalLB}
	In this section we show that size of $k$-WTSS of a graph can be of size at least $\Omega(n^2)$ even for $k=1$ if we allow either weight function to be rational valued or increment function to be rational valued.
	\subsection{Lower bound for rational valued weight function}
	\label{sec:rationalweightLB}
	\begin{theorem}
	\label{thm:generalLB1}
	If weight of an edge can be any rational value, then for every $n \in \mathbb{N}$, there exists a directed graph with $n$ vertices whose $1$-WTSS must contain $c\cdot n^2$ many edges for some constant $c >0$.
	\end{theorem}
	\begin{proof}
	Take any $n\in \mathbb{N}$. Now consider a graph $G$ with the vertex set $V(G)=\{s,v_1,v_2,\cdots,v_{n-1}\}$ and following edge set
	$$E(G)=\{(s,v_1)\}\cup \{(v_i,v_j) \mid i,j\in [n-1] \text{ and }i < j\}.$$
	So, $|E(G)|={{n-1}\choose {2}}+1$. Next define the weight function $w:E(G) \to \mathbb{Q}$ as follows:
	$$w(e)=\begin{cases}1-\frac{i+j}{2n} \quad &\text{if $e=(v_i,v_j)$ and $j \ne i+1$}\\
	n&\text{if $e=(s,v_1)$}\\
	0&\text{otherwise} \end{cases}$$
	Now we show that any $1$-WTSS of $G$ must contain all the edges in $E(G)$. Note that initially for any vertex $v_i$, the shortest $s-v_i$ path has weight $n$. As $(s,v_1)$ and $(v_{j-1},v_{j})$'s for $j \le i$ must lie on any shortest $s-v_i$ path so they must belong to any $1$-WTSS. Let us now consider any edge $(v_i,v_j)$ such that $j \ne i+1$. Take an increment function $I:E(G) \to \mathbb{N}$ such that $I((v_i,v_{i+1}))=1$ and for all other edges $e \ne (v_i,v_{i+1})$, $I(e)=0$. It is easy to see that the unique shortest path under the new weight function $w'$ defined by $w'(e)=w(e)+I(e)$ for all $e \in E(G)$ follows the original shortest path till $v_i$ from $s$ and then take the edge $(v_i,v_j)$. Thus $(v_i,v_j)$ must belong to any $1$-WTSS of $G$ and this concludes the proof.
	\end{proof}
	Reader may note that the choice of weight $n$ on the edge $(s,v_1)$ in the above mentioned construction is arbitrary and one is free to take any value instead of $n$. 
	\subsection{Lower bound for rational valued increment function}
	\label{sec:rationalincrementLB}
	\begin{theorem}
	\label{thm:generalLB2}
	If it is allowed to increase the weight of the edges by any rational value, then for every $n \in \mathbb{N}$, there exists a directed graph with $n$ vertices whose $1$-WTSS must contain $c\cdot n^2$ many edges for some constant $c >0$.
	\end{theorem}
	\begin{proof}
	Consider any $n \in \mathbb{N}$. Consider two sets of vertices $A$ and $B$ such that $|A|=\lfloor \frac{n}{2} \rfloor-1$ and $|B| = n-\lfloor \frac{n}{2} \rfloor - 1$. Now take a graph $G$ with $V(G) = \{s,t\} \cup A \cup B$. Then connect $s$ with all the vertices in $A$ and all the vertices in $B$ with $t$. Next add an edge between any vertex in $A$ and any vertex in $B$. So we get
	$$E(G)=\{(s,v) \mid  v \in A\} \cup \{(u,v) \mid u \in A \text{ and } v \in B\} \cup \{(v,t) \mid  v \in B\}.$$
	Hence $|E(G)| \ge c \cdot n^2$ for some constant $c > 0$. Now consider a weight function $w:E(G) \to \mathbb{Z}$ such that for all $e \in E(G)$, $w(e)=1$. Note that any $s-t$ path has weight $3$. Now we show that every edge in this graph must be present in its $1$-WTSS. First observe that every $s-t$ path in $G$ is of weight $3$ and every edge $e \in E(G)$ belong to some $s-t$ path. So it is sufficient to show that all the $s-t$ paths in $G$ must be included in its $1$-WTSS.
	
	Let $P$ be one such $s-t$ path and suppose it passes through a vertex $u \in A$ and $v \in B$. Next consider a increment function $I:E(G) \to \mathbb{Q}$ such that 
	$$I(e)=\begin{cases}\frac{1}{|A|-1} \quad &\text{if $e=(s,x)$ and $x \ne u$}\\
	\frac{1}{|B|-1}&\text{if $e=(y,t)$ and $y \ne v$}\\
	0&\text{otherwise} \end{cases}$$
	Clearly, the above increment function $I$ satisfies the condition that $\sum_{e \in E(G)}I(e) \le 1$ and all the $s-t$ paths except $P$ has now weight strictly greater than $3$ under the new weight function $w'$ defined by $w'(e)=w(e)+I(e)$ for all $e \in E(G)$. Hence all the edges of path $P$ must be included in any $1$-WTSS of $G$.
	\end{proof}

	\paragraph*{Remarks:} We emphasize that all the lower bound results in the above section hold for undirected graphs also. Moreover, exactly the same graphs without any direction will serve the purpose.
	
	\section{Discussion}
	\label{sec:conclusion}
	In this paper we initiate the study of single source shortest path problem in a model where weight of any edge can be increased. This model is motivated from congestion in any network and is simpler than the edge fault model. To summarize, we have provided an efficient algorithm to compute a sparse subgraph that preserves the distances from any designated source vertex and is also resilient under bounded weight increment. When the weight increment is bounded by $k$ then the subgraph computed by our algorithm will be of size at most $O(k^kn)$. We also show a lower bound of $\frac{5}{4}2^k n$ on the size of such a subgraph. This shows that our construction is tight upto some constant as long as $k$ is bounded by some constant. Though, it is interesting to farther study this problem to close the gap between the upper and the lower bound and at this point we would like to leave this problem as open. Another open problem is to improve the run time of the construction.
	
	We have already shown in this paper that from the perspective of constructing sparse distance preserver, the weight tolerant model is much simpler than fault tolerant model. It might be possible that the same is true for problems like finding distance oracle in this weight increment model. As a corollary of our result one can get an $O_k(n)$ space and $O_k(n^2)$ time data structure that will answer the distance queries from a single source vertex, where $O_k(\cdot)$ notation denotes involvement of a constant that depends only on $k$. It is interesting to farther study this problem to reduce the query time to $O_k(1)$ while getting some reasonable bound on the space requirement.

	\section*{Acknowledgments}
	The first author would like to thank Pavan Aduri and Vinodchandran N. Variyam for some helpful discussions during initial phase of this work and a special thank to Pavan Aduri for suggesting to study this problem. Authors also thank Keerti Choudhary and Michal Kouck{\' y} for many valuable suggestions and comments.

	\bibliographystyle{plainurl}
	\bibliography{faultShort}
	\appendix
	
	\section{Explanations for Assumption~\ref{asm:degBound}}
	\label{app:assumption}
	We are given a directed graph $G=(V,E)$ and the associated weight function $w:E(G)\to \mathbb{R}$ as input. We claim that we can construct a new graph $G'=(V',E')$ with a weight function $w':E(G')\to \mathbb{R}$ such that the out-degree of every vertex $v'\in V'$ is bounded by $2$. Now we describe how to get such a $G'$ from $G$. First, for each $ v\in V$, we construct a binary tree $T_v$ as follows: Define $r_v$ to be the root of $T_v$. Suppose $d(v)$ denotes the out-degree of $v$ in $G$. Then $T_v$ contains exactly $d(v)$ many leaves, say $l_v^1, \cdots, l_v^{d(v)}$. Let $(v,u_1), \cdots, (v,u_{d(v)})$ are the out edges of $v$ in $G$. We delete all of them and in that place we insert the binary tree $T_v$ by adding an edge from vertex $v$ to $r_v$ and adding edges from vertex $l_v^i$ to $u_i$, for all $i\in [d(v)]$. We define the weight function $w'$ for $G'$ as follows: Set, $w'(v,r_v)=0$, $w'(l_v^i, u_i)=w(v,u_i)$ for all $i\in [d(v)]$ and for rest of the edges $e\in T_v$ set $w'(e)=0$. Subsequently we observe the following properties of $G'$. 
	\begin{enumerate}
	\item Every vertex of $G'$ has out-degree at most $2$ whereas in-degree is same as that of in the graph $G$.
	\item Graph $G'$ has $O(m)$ vertices and $O(m)$ edges. 
	\item Every edge $(v,u_i)$ of $G$ is represented by a path $P_{v \to u_i}=(v,r_v)\circ$(path from $r_v$ to $l_v^i$ in $T_v$)$\circ (l_v^i, u_i)$ in $G'$ and $w(v,u_i)=w'(P_{v \to u_i})$. Hence for any vertex $v\in V(G)$, $dist_{G,w}(s,v)=dist_{G',w'}(s,v)$. 
	\end{enumerate} 
	Now we show that given a $k$-WTSS($t$), say $H'$ of the graph $G'$ how to construct a $k$-WTSS($t$), say $H$ for the graph $G$. We build $H$ as follows: For each $u_i\in Out(v)$ of the graph $G$, we include an edge $(v,u_i)$ in $H$ if and only if the edge $(l_v^i, u_i)$ is present in graph $H'$. Now the claim is that $H$ is a $k$-WTSS($t$) for the graph $G$. Let $I:E(G) \to \mathbb{Z}$ be any increment function on graph $G$ such that $\sum_{e \in E(G)}I(e) \le k$. Now define another increment function $I':E(G') \to \mathbb{Z}$ for the graph $G'$ as follows: For every edge $(v,u_i)$, set $I'(l_v^i,u_i)=I(v,u_i)$. For all other edges $e\in G'$, set $I'(e)=0$. Clearly, 
	$$\sum_{e \in E(G')}I'(e)=\sum_{e \in E(G)}I(e) \le k.$$ 
	Now from the construction we can observe that for any vertex $t$ in $G$, $dist_{H,w+I}(s,t)=dist_{H',w'+I'}(s,t)$. Now as $H'$ is a $k$-WTSS($t$) for the graph $G'$, we have that $dist_{H',w'+I'}(s,t)=dist_{G',w'+I'}(s,t)$ and therefore 
	$$dist_{H,w+I}(s,t)=dist_{H',w'+I'}(s,t)=dist_{G',w'+I'}(s,t)=dist_{G,w+I}(s,t).$$ 
	Hence $H$ is a $k$-WTSS($t$) for the graph $G$ and in-degree of any vertex in $H$ is same as that of in $H'$. Hence for any vertex $t\in V(G)$, computing $k$-WTSS($t$) of $G$ is same as computing $k$-WTSS($t$) for $G'$ which has $O(m)$ vertices and edges, and out-degree of every vertex is bounded by two.

	\section{Lower Bound for Weight Decrement Model}
	\label{app:decrementLB}
	One can relax our model by allowing decrement operation on edge weight. Unfortunately, one can easily show that there are graphs for which there is no sub-quadratic sized subgraph that preserves the distances from a single source under this relaxed model.
	\begin{theorem}
	If it is allowed to decrement the weight of the edges even by a integer, then for every $n \in \mathbb{N}$, there exists a directed graph with $n$ vertices whose $1$-WTSS must contain $c\cdot n^2$ many edges for some constant $c >0$.
	\end{theorem}
	\begin{proof}
	Consider any $n \in \mathbb{N}$. Consider two sets of vertices $A$ and $B$ such that $|A|=\lfloor \frac{n}{2} \rfloor-1$ and $|B| = n-\lfloor \frac{n}{2} \rfloor - 1$. Now take a graph $G$ with $V(G) = \{s,t\} \cup A \cup B$. Then connect $s$ with all the vertices in $A$ and all the vertices in $B$ with $t$. Next add an edge between any vertex in $A$ and any vertex in $B$. So we get
	$$E(G)=\{(s,v) \mid  v \in A\} \cup \{(u,v) \mid u \in A \text{ and } v \in B\} \cup \{(v,t) \mid  v \in B\}.$$
	Hence $|E(G)| \ge c \cdot n^2$ for some constant $c > 0$. Now consider a weight function $w:E(G) \to \mathbb{Z}$ such that for all $e \in E(G)$, $w(e)=2$. Note that every $s-t$ path has weight $6$. Now we show that every edge in this graph must be present in its $1$-WTSS. Let us take any edge $e\in E(G)$ and decrement its weight by $1$. It is now easy to see that any $s-t$ path that passes through $e$ has weight $5$ whereas all other $s-t$ paths still have weight $6$. Hence the edge $e$ must be included in any $1$-WTSS of $G$.
	\end{proof}

\end{document}

%% file: macros.tex
\usepackage[]{algorithm2e}

\newtheorem{assumption}{Assumption}

\newenvironment{proof-sketch}{\noindent{\bf Sketch of Proof}\hspace*{1em}}{\qed\bigskip}

\floatstyle{ruled}
\newfloat{algorithm}{tbp}{loa}
\providecommand{\algorithmname}{Algorithm}
\floatname{algorithm}{\protect\algorithmname}

\floatstyle{ruled}
\newfloat{algorithm}{tbp}{loa}
\providecommand{\algorithmname}{Algorithm}
\floatname{algorithm}{\protect\algorithmname}

\theoremstyle{plain}
\makeatother

%% file: style.tex
\usepackage{epsfig}
\usepackage{url}

\newtheorem{theorem}{Theorem}

\newtheorem{lemma}{Lemma}[section]
\newtheorem{claim}{Claim}[section]

\newtheorem{observation}{Observation}[section]
\newtheorem{corollary}{Corollary}[section]

\newtheorem{definition}{Definition}

\newcommand{\ignore}[1]{}

\newcommand{\newfontobj}[2]{
  \newcommand{#1}[1]{
    \expandafter\def\csname##1\endcsname{{#2 ##1}}}}

\comment{

\class{RL} \class{BPL}

\class{PSPACE}          
\class{AM}              
\class{AMTIME}
\class{NC}
\class{AMSUBEXP}
\class{MA}
\class{NP}
\class{NLIN}
\class{L}
\class{BPSPACE}
\class{LINSPACE}
\class{P}
\class{RP}
\class{BPP}
\class{RNC}
\class{PrAM}
\class{NPSV}
\class{advBPP}
\class{DTIME}
\class{E}
\class{ZPP}
\class{EXPSPACE}
\class{SUBEXP}
\class{PH}
\class{IP}
\class{SNP}
\class{SNQP}
\class{AVP}
\class{QP}
\class{SNTIME}
\class{DSPACE}
\class{BPSPACE}
\class{BPTIME}

\class{DTIME}

\class{AvgTIME} \class{DistNP} \class{AZPP} \class{AVE}
\class{DistE} \class{DistEXP} \class{E} \class{NE} \class{UE}
\class{EXP} \class{pcomp} \class{NEXP} \class{UP} \class{DTIME}
\class{NTIME} \class{R} \class{NSUBEXP} \class{SUBEXP} \class{PF}
\class{LINcomp}
\class{BPTIME}
\class{DistNLIN}
\class{AME}
\class{SIZE}
\class{NSIZE}
\class{SAT}
\class{SIZE}

}

%% file: auxiliary.tex
\newcommand{\convexpath}[2]{
[   
    create hullnodes/.code={
        \global\edef\namelist{#1}
        \foreach [count=\counter] \nodename in \namelist {
            \global\edef\numberofnodes{\counter}
            \node at (\nodename) [draw=none,name=hullnode\counter] {};
        }
        \node at (hullnode\numberofnodes) [name=hullnode0,draw=none] {};
        \pgfmathtruncatemacro\lastnumber{\numberofnodes+1}
        \node at (hullnode1) [name=hullnode\lastnumber,draw=none] {};
    },
    create hullnodes
]
($(hullnode1)!#2!-90:(hullnode0)$)
\foreach [
    evaluate=\currentnode as \previousnode using \currentnode-1,
    evaluate=\currentnode as \nextnode using \currentnode+1
    ] \currentnode in {1,...,\numberofnodes} {
-- ($(hullnode\currentnode)!#2!-90:(hullnode\previousnode)$)
  let \p1 = ($(hullnode\currentnode)!#2!-90:(hullnode\previousnode) - (hullnode\currentnode)$),
    \n1 = {atan2(\x1,\y1)},
    \p2 = ($(hullnode\currentnode)!#2!90:(hullnode\nextnode) - (hullnode\currentnode)$),
    \n2 = {atan2(\x2,\y2)},
    \n{delta} = {-Mod(\n1-\n2,360)}
  in 
    {arc [start angle=\n1, delta angle=\n{delta}, radius=#2]}
}
-- cycle
}

\begin{figure}
\centering

\begin{tikzpicture}[thick,every node/.style={draw,circle,fill=white}]

\begin{scope}[level distance=1cm]
\node[draw,circle](1){s}
 child{
	node[draw,circle](2){$v_1$}
	child{ 
	    node[draw,circle,left=.46cm](11){}
		child{ 
	    node[draw,circle](12){}
	} 
	child{ 
	    node[draw,circle](13){}
	}
	} 
	child{ 
	    node[draw,circle](3){}
	    child{
		node[draw,circle](5){}
		child{ 
	    node[draw,circle](7){}
	    child{ 
	    node[draw,circle](14){}
	    child{ 
	    node[draw,circle](15){t}
	}
	}
	} 
	child{ 
	    node[draw,circle](8){}
	} 
	    }
	} 
	child{ 
	    node[draw,circle](4){$u_2$}
	    child{
		node[draw,circle,right=.46cm](6){$v_2$}
		child{ 
	    node[draw,circle](9){}
	} 
	child{ 
	    node[draw,circle](10){}
	}
	    }
	} 
};
\draw[blue,->,ultra thick] (1)--(2);
\draw[white,->] (2)--(11);
\draw[dashed,->] (2)--(11);
\draw[->,ultra thick] (2)--(3);
\draw[->,ultra thick] (2)--(4);
\draw[red,->, ultra thick] (3)--(5);
\draw[red,->, ultra thick] (4)--(6);
\draw[->,ultra thick] (5)--(7);
\draw[white,->] (5)--(8);
\draw[dashed,->] (5)--(8);
\draw[->,ultra thick] (6)--(9);
\draw[->,ultra thick] (6)--(10);
\draw[white,->] (11)--(12);
\draw[dashed,->] (11)--(12);
\draw[white,->] (11)--(13);
\draw[dashed,->] (11)--(13);
\draw[white,->] (13)--(7);
\draw[dashed,->] (13)--(7);
\draw[->,ultra thick] (7)--(14);
\draw[green,->,ultra thick] (14)--(15);
\draw[white,->] (12)--(15);
\draw[dashed,->] (12)--(15);
\draw[green,->,ultra thick] (10)--(15);
\draw[green,->,ultra thick] (9)--(15);
\draw[white,->] (8)--(15);
\draw[green,dashed,->] (8)--(15);
\draw[->,>=stealth,brown,ultra thick] (1)[out=30,in=30] to (2);
\draw[->,>=stealth,brown,ultra thick] (1)[out=30,in=30] to (6);
\begin{pgfonlayer}{background}
\fill[yellow,opacity=0.2] \convexpath{1,2,3,5,7,14,15,10,1}{0.5pt};
\end{pgfonlayer}
\end{scope}


\end{tikzpicture}

\caption{Suppose the yellow colored region represents $G_{short}$. The edges of $C_{1}$ and $C_{2}$ are colored with blue and red respectively. Brown colored edges are the edges added in the auxiliary graph $G'$ and the edges colored with green constitute the set $\mathcal{E}(t)$. Paths represented by the thick edges are the $3$ edge disjoint paths in $G'$ when $r=2$.}
   \label{fig:auxiliary}
   \end{figure}
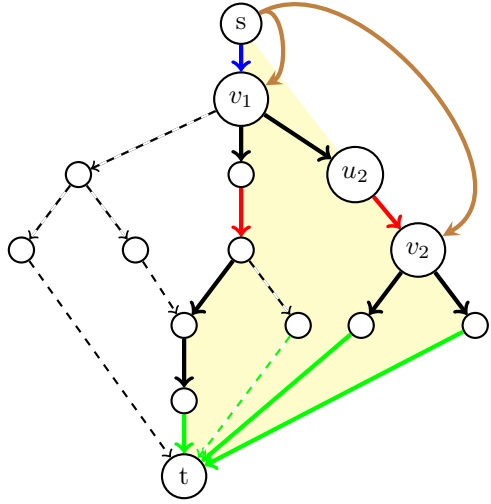

%% file: example.tex
\newcommand{\convexpath}[2]{
[   
    create hullnodes/.code={
        \global\edef\namelist{#1}
        \foreach [count=\counter] \nodename in \namelist {
            \global\edef\numberofnodes{\counter}
            \node at (\nodename) [draw=none,name=hullnode\counter] {};
        }
        \node at (hullnode\numberofnodes) [name=hullnode0,draw=none] {};
        \pgfmathtruncatemacro\lastnumber{\numberofnodes+1}
        \node at (hullnode1) [name=hullnode\lastnumber,draw=none] {};
    },
    create hullnodes
]
($(hullnode1)!#2!-90:(hullnode0)$)
\foreach [
    evaluate=\currentnode as \previousnode using \currentnode-1,
    evaluate=\currentnode as \nextnode using \currentnode+1
    ] \currentnode in {1,...,\numberofnodes} {
-- ($(hullnode\currentnode)!#2!-90:(hullnode\previousnode)$)
  let \p1 = ($(hullnode\currentnode)!#2!-90:(hullnode\previousnode) - (hullnode\currentnode)$),
    \n1 = {atan2(\x1,\y1)},
    \p2 = ($(hullnode\currentnode)!#2!90:(hullnode\nextnode) - (hullnode\currentnode)$),
    \n2 = {atan2(\x2,\y2)},
    \n{delta} = {-Mod(\n1-\n2,360)}
  in 
    {arc [start angle=\n1, delta angle=\n{delta}, radius=#2]}
}
-- cycle
}

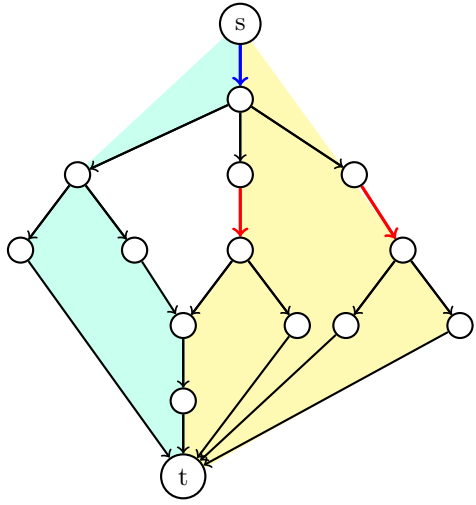
\begin{figure}
\centering

\begin{tikzpicture}[thick,every node/.style={draw,circle,fill=white}]

\begin{scope}[level distance=1cm]
\node[draw,circle](1){s}
 child{
	node[draw,circle](2){}
	child{ 
	    node[draw,circle,left=.46cm](11){}
		child{ 
	    node[draw,circle](12){}
	} 
	child{ 
	    node[draw,circle](13){}
	}
	} 
	child{ 
	    node[draw,circle](3){}
	    child{
		node[draw,circle](5){}
		child{ 
	    node[draw,circle](7){}
	    child{ 
	    node[draw,circle](14){}
	    child{ 
	    node[draw,circle](15){t}
	}
	}
	} 
	child{ 
	    node[draw,circle](8){}
	} 
	    }
	} 
	child{ 
	    node[draw,circle](4){}
	    child{
		node[draw,circle,right=.46cm](6){}
		child{ 
	    node[draw,circle](9){}
	} 
	child{ 
	    node[draw,circle](10){}
	}
	    }
	} 
};
\draw[blue,->,very thick] (1)--(2);
\draw[->] (2)--(11);
\draw[->] (2)--(3);
\draw[->] (2)--(4);
\draw[red,->, very thick] (3)--(5);
\draw[red,->, very thick] (4)--(6);
\draw[->] (5)--(7);
\draw[->] (5)--(8);
\draw[->] (6)--(9);
\draw[->] (6)--(10);
\draw[->] (11)--(12);
\draw[->] (11)--(13);
\draw[->] (13)--(7);
\draw[->] (7)--(14);
\draw[->] (14)--(15);
\draw[->] (12)--(15);
\draw[->] (10)--(15);
\draw[->] (9)--(15);
\draw[->] (8)--(15);
\begin{pgfonlayer}{background}
\fill[yellow,opacity=0.3] \convexpath{1,2,3,5,7,14,15,10,1}{0.5pt};
\fill[green!30!cyan!70,opacity=0.3] \convexpath{1,2,11,13,7,14,15,12,11,1}{0.5pt};
\end{pgfonlayer}
\end{scope}


\end{tikzpicture}

\caption{Region shaded with green color represents $G_{\sigma}$ for $\sigma=(1,-1,\cdots,-1)$ whereas yellow colored region is the shortest path subgraph of $G$. The edges of $C_{(-1,\cdots,-1),1}$ and $C_{(-1,\cdots,-1),2}$ are colored with blue and red respectively. $G_{\sigma}$ is obtained by removing red colored edges.}
   \label{fig:example}
   \end{figure}